\documentclass[a4paper,10pt,openright]{paper}
\usepackage[english]{babel}
\usepackage[latin1]{inputenc}
\usepackage{amsmath,amsthm,amsfonts,amscd,eucal,latexsym,amssymb,mathrsfs,cancel,tikz}
\usepackage{epsfig}
\usepackage{simplewick}
\usepackage{hyperref}
\usepackage{color}

\definecolor{NiColor}{RGB}{77,77,255}
\definecolor{NiColoRed}{RGB}{255,77,77}
\definecolor{NiCitation}{RGB}{77,255,77}

\newtheoremstyle{TheoremStyle}
        {3pt}
        {3pt}
        {}
        {}
        {\bf}
        {:}
        {.5em}
        {}
\newtheoremstyle{ExampleAndRemarkStyle}
        {3pt}
        {3pt}
        {}
        {}
        {\bf}
        {:}
        {.5em}
        {}
\newtheoremstyle{ProofStyle}
        {3pt}
        {3pt}
		{}
        {}
        {\bf}
        {:}
        {.5em}
        {}

\theoremstyle{TheoremStyle}
\newtheorem{theorem}{Theorem}
\newtheorem{corollary}[theorem]{Corollary}

\newtheorem{lemma}[theorem]{Lemma}
\newtheorem{Definition}[theorem]{Definition}
\theoremstyle{ExampleAndRemarkStyle}
\newtheorem{remark}[theorem]{Remark}
\theoremstyle{ProofStyle}

\title{Thermal state with quadratic interaction}
\date{}
\author{Nicolò Drago}

\begin{document}
\maketitle
\begin{flushleft}
\begin{small}
Dipartimento di Fisica, Università degli Studi di Pavia, Via Bassi, 6, I-27100 Pavia, Italy.\\
Istituto Nazionale di Fisica Nucleare -- Sezione di Pavia, Via Bassi, 6, I-27100 Pavia, Italy.\\
Istituto Nazionale d'Alta Matematica ``Francesco Severi'', Unità di Ricerca dell'Università di Pavia, Dipartimento di Matematica -- Via Adolfo Ferrata, 5, 27100 Pavia PV, Italy.\\
mail: nicolo.drago@unipv.it
\end{small}
\end{flushleft}

\begin{abstract}
We consider the perturbative construction, proposed in \cite{FrLi14}, for a thermal state $\Omega_{\beta,\lambda V\{f\}}$ for the theory of a real scalar Klein-Gordon field $\phi$ with interacting potential $V\{f\}$.
Here $f$ is a spacetime cut-off of the interaction $V$ and $\lambda$ is a perturbative parameter.
We assume that $V$ is quadratic in the field $\phi$ and we compute the adiabatic limit $f\to 1$ of the state $\Omega_{\beta,\lambda V\{f\}}$.
The limit is shown to exist, moreover, the perturbative series in $\lambda$ sums up to the thermal state for the corresponding (free) theory with potential $V$.
In addition, we exploit the same methods to address a similar computation for the non-equilibrium steady state (NESS) \cite{Ru00} recently constructed in \cite{DrFaPi17}.
\end{abstract}

\section{Introduction}
Algebraic quantum field theory (AQFT) is a mathematically rigorous approach to quantum field theory (QFT).
Nowadays, AQFT is a well-established set-up to describe the propagation of quantum fields on curved spacetimes \cite{BrDaFrYn15,Wa94}.

The approach can be summarized as follows.
To any physical system one associates a $*$-algebra $\mathcal{A}$, whose elements are interpreted as the observables of the system.
Algebraic relations reproduce natural assumptions on the structural properties of the observables, while the whole construction is subjected to the requirement of covariance \cite{BrFrVe03,Kay92}, which ensures that the $*$-algebra $\mathcal{A}$ is coherently constructed on any globally hyperbolic spacetime \cite{BeSa05,BeSa06}.
The dynamics can be implemented algebraically through the time-slice axiom \cite{ChFr09}.
Once the algebra $\mathcal{A}$ has been identified, the notion of state can be introduced \cite[Chap.5]{BrDaFrYn15}.
The latter is, per definition, a linear, positive and normalized functional $\Omega\colon\mathcal{A}\to\mathbb{C}$.
Yet, not all states are found to be physically relevant: a necessary constraint is the so-called Hadamard condition \cite{FeVe13,FuSwWa78,FuNaWa81,Wa94} which has been recast in the framework of Microlocal Analysis \cite{Ho83} in the seminal works \cite{Ra96,RaVe96}.

The algebraic approach has been successfully applied and it is well-understood for free theories \cite{BrDaFrYn15}.
Interacting theories can be addressed with the same techniques, however, the underlying non-linearity of the equations of motion creates additional difficulties.
To handle this problem, one usually switches to the perturbative approach, which can be described as follows.
Since the non-linear dynamics can be read as a correction $V$ of a linear dynamics, one can try to expand interacting observables as formal power series in a formal parameter $\lambda$ for the interacting potential $\lambda V$.
From a technical point of view, the above mentioned expansion of interacting observables is realized through the so-called quantum M\o ller operator $\mathsf{R}_{\lambda V}$ \cite{BaFr09,FrRe14,Pe52}.
This operator is defined through the famous Bogoliubov' formula, which requires the introduction of Wick polynomials and of the time-ordered product \cite{BrFr00,EpGl73,HoWa01,HoWa02,HoWa05,KhMo14,Khavkine-Melati-Moretti-17} see also \cite[Chap.2]{BrDaFrYn15}.
Once an extension of the time-ordered product has been fixed, the algebra $\mathcal{A}_V$ of interacting observables can be defined with the quantum M\o ller operator $\mathsf{R}_{\lambda V}$ as a $*$-subalgebra of the algebra $\mathcal{A}[[\lambda]]$ of formal power series in $\lambda$ with values in $\mathcal{A}$.
The resulting algebra $\mathcal{A}_V$ satisfies the condition of covariance \cite{BrFrVe03,Kay92} as well as the time-slice axiom \cite{ChFr09}.
The whole construction applies assuming that the perturbation $V$ is itself an element of the algebra $\mathcal{A}$.
The procedure of removing the compactness in the support of $V$ is known under the name of adiabatic limit.
The latter has been implemented algebraically \cite[Chap.1-2]{BrDaFrYn15} and it has recently been improved thanks to the results of \cite{BaRe16,HaRe16}.
In particular in \cite{BaRe16} the authors have shown convergence of the expectation value of the quantum M\o ller operator in the case of the Sine-Gordon Model.

The problem of identifying physically interesting states on the interacting algebra $\mathcal{A}_V$ has been addressed recently in \cite{FrLi14}.
Therein, the authors successfully applied a construction proposed in \cite{Ar73,BrRo97a,BrRo97b} in the framework of $C^*$-algebras.
The latter allows to construct a thermal equilibrium state $\Omega_{\beta,\lambda V}$ \cite{BrKiRo78,HaHuWi67,HaKaPo74,Ro73,SaVe00} for the interacting theory once a corresponding thermal state $\Omega_\beta$ for the free theory has been given.
In \cite{FrLi14} the construction of $\Omega_{\beta,\lambda V}$ has been achieved in terms of a formal power series in $\lambda$ exploiting the time-slice axiom \cite{ChFr09} of the algebra $\mathcal{A}_V$.
Further results on this state can be found in \cite{DrFaPi17,Drago-Faldino-Pinamonti-18,DrHaPi16,Hack-Verch-2018}.

In this paper we analyse the state $\Omega_{\beta,\lambda V}$ for the case of a quadratic potential $V$.
In this particular case the perturbed free theory leads to another free theory.
This perturbation may model a variation in the mass term of the Klein-Gordon operator $\square+m^2\to\square+m^2+\lambda m_0^2$  -- with $m^2>0,m^2+\lambda m_0^2>0$ -- though more general situations are allowed.
The assumption on $V$ allows to investigate the adiabatic limit of the resulting state $\Omega_{\beta,\lambda V}$, which is computed order-by-order.
The series for the resulting state can be evaluated directly and it is shown to lead to the corresponding state associated with the perturbed free theory, see Theorem \ref{Theorem: main theorem}.

The convergence of the state $\Omega_{\beta,\lambda V}$ in the case of a quadratic potential $V$ is expected -- see for example \cite{Derezinski-Jaksic-Pillet-03,Donald-90,Sakai-87} where non-bounded perturbations of KMS were considered -- however, the results and the tools exploited in proving the main result are noteworthy for several reasons.
First of all, this computation shows that perturbation theory is reliable: The adiabatic limit can be taken order-by-order, leading to a series which sums up to the correct result.
This behaviour is expected but a priori not guaranteed and this result increases the chances of perturbation theory of being the correct approach to interacting theories.

The second remarkable point of this analysis is that the tools used in the proof of the main result can potentially be generalized to a generic non-linear potential $V$.
In particular, the first bit of information exploited in the computation of the adiabatic limit is the possibility to interchange the quantum M\o ller operator $\mathsf{R}_{\lambda V}$ with its classical counterpart $\mathsf{R}^{\textrm{cl}}_{\lambda V}$ \cite{Pe52} -- see equation \eqref{Equation: PPA}.
From a computational point of view, this leads to a great simplification, due to the results of \cite{DaDr16,DrGe16,DrHaPi16}.
From an abstract point of view, equation \eqref{Equation: PPA} can be understood as an effective resummation of the perturbative series and it should be compared to other approaches \cite{Al90,As07,LeB00,St71,St95}.
It would be extremely interesting to understand to which extent equation \eqref{Equation: PPA} can be generalized to non-linear potential $V$.
Most likely, this would allow to interchange the quantum M\o ller operator $\mathsf{R}_{\lambda V}$ with a classical one $\mathsf{R}^{\textrm{cl}}_{\lambda V_{\textrm{eff}}}$, with $V_{\textrm{eff}}$ being an effective potential built out of $V$ \cite{BrDu08,We96}.
From this point of view, the results of this paper can be understood as a promising starting point for an ``effective analysis'' of perturbative AQFT (pAQFT).

Finally, this result points towards a non-perturbative version of pAQFT, whose first steps will be necessarily based on a systematic check of the convergence of the perturbative approach, in the spirit of \cite{BaRe16}.

The paper is organized as follows: In section \ref{Section: Brief resume in pAQFT} we briefly summarize the functional approach to perturbative algebraic quantum field theory (pAQFT) for the Klein-Gordon field on Minkowski spacetime as well as the construction proposed in \cite{FrLi14}.
Section \ref{Section: Main result} contains the main result of the paper, see Theorem \ref{Theorem: main theorem}, which is proved in section \ref{Section: Proof of Main theorem}.
Finally, in section \ref{Section: NESS}, the techniques developed in the previous sections are applied to the non-equilibrium steady state (NESS) \cite{Ru00} constructed in \cite{DrFaPi17}.

\section{Brief resumé of pAQFT}\label{Section: Brief resume in pAQFT}
In this section we give a brief introduction to the quantization of the real scalar Klein-Gordon field in the framework of algebraic quantum field theory \cite{BaFr09,BrFrVe03,DuFr01,FrRe13,FrRe14,Kay92}, see also \cite[Chap.2]{BrDaFrYn15}.
This approach applies on any globally hyperbolic spacetime \cite{BaGi12,BeSa05,BeSa06}, and it is covariant in the sense of a generally covariant local theory introduced in \cite{BrFrVe03,Kay92} see also \cite{HoWa01}.
For practical purposes we focus our attention to Minkowski spacetime $M$, because the results of \cite{FrLi14} were developed on this particular background.
The main reference for this section is \cite[Chap.2]{BrDaFrYn15}.

\subsection{Free theory}
In this section we outline the quantization of a \textit{free} real scalar Klein-Gordon field $\phi$, whose dynamics is ruled by the massive Klein-Gordon equation $\square\phi+m^2\phi=0$, $m>0$, $\square=-\eta^{ab}\partial_a\partial_b$ where $\eta=\textrm{diag}(-1,1,1,1)$ -- we exploit natural units $\hbar=c=1$.
We will consider the functional approach \cite{DuFr01}, where the (off-shell) $*$-algebra of observables is identified as that of functionals over kinematic configurations $\phi\in C^\infty(M)$, namely $F\colon C^\infty(M)\to\mathbb{C}$.
For the sake of simplicity, we focus on \textit{polynomial} functionals $F\in\mathcal{P}$, which lead to some simplification without spoiling the full generality of this approach.
Notice that a polynomial functional $F$ is automatically smooth, that is, for all $\phi,\psi\in C^\infty(M)$ the function $\mathbb{R}\ni x\mapsto F(\phi+x\psi)$ is differentiable at $x=0$ and, for all natural numbers $n\geq 1$, its $n$-th derivative at $x=0$ defines a symmetric distribution, denoted $F^{(n)}[\phi]$ and called the $n$-th functional derivative of $F$ at $\phi$.
Explicitly $F(\phi+x\psi)^{(n)}\big|_{x=0}=F^{(n)}[\phi](\psi^{\otimes n})$.
Unless stated otherwise, from now on all functionals will be implicitly considered to be polynomial.

Among all, local functionals will play an important r\^ole in the construction of the algebra of free observables.
A functional $F\colon C^\infty(M)\to\mathbb{C}$ is said to be local if it satisfies the two following conditions:
(i)
$F$ is compactly supported, that is $\textrm{spt}(F):=\overline{\bigcup_{\phi\in C^\infty(M)}\textrm{spt}\big(F^{(1)}[\phi]}\big)$ is compact;
(ii)
for all $n\geq 1$ and $\phi\in C^\infty(M)$, the $n$-th functional derivative of $F$ at $\phi$ is supported on the full diagonal of $M^n$, that is
$\textrm{spt}(F^{(n)}[\phi])\subseteq\{(x_1,\ldots,x_n)\in M^n|\; x_1=\ldots=x_n\}$.
The set of local functionals will be denoted by $\mathcal{P}_{\textrm{loc}}$.

Once equipped with the pointwise product, the set $\mathcal{P}_{\textrm{loc}}$ generates the algebra $\mathcal{P}_{\textrm{mloc}}$ of \textit{multilocal} functionals.
Together with the $*$-involution defined by the complex conjugation $F^*(\phi):=\overline{F(\phi)}$, one obtains a commutative $*$-algebra, identified with that of classical observables for the Klein-Gordon field.
In order to introduce its quantum counterpart, one needs to deform the pointwise product of $\mathcal{P}_{\textrm{mloc}}$.
This is realized by choosing a so-called Hadamard distribution $\omega$ \cite{Ra96,RaVe96}, which is defined as 
a positive distribution $\omega\in C^\infty_{\textrm{c}}(M^2)'$ which  satisfies the canonical commutation relations (CCR), that is $\omega(f,g)-\omega(g,f)=i\mathsf{G}(f,g)$.
Here $\mathsf{G}$ denotes the causal propagator \cite{BaGi12} associated to $\square+m^2$.
Moreover, the Wave Front Set \cite{Ho83} of the distribution $\omega$ is required to satisfy the microlocal spectrum condition \cite{Ra96,RaVe96,Wa94} -- see also equation \eqref{Equation: general momentum expansion of Hadamard distribution}.
The latter requirement ensures that the singular behaviour of $\omega$ is the same as that of the Minkowski vacuum.

Once an Hadamard distribution has been chosen one may define an associative, non-commutative, $\star$-product on $\mathcal{P}_{\textrm{mloc}}$ as follows \cite{BaFr09,Di90,Di92,DuFr01}: for all $F,G\in\mathcal{P}_{\textrm{mloc}}$ one sets
\begin{gather}\label{Equation: definition of star-product}
\big(F\star_\omega G\big)(\phi):=
F(\phi)G(\phi)+
\sum_{n\geq 1}
\frac{1}{n!}
\omega^{\otimes n}\bigg(
F^{(n)}[\phi],
G^{(n)}[\phi]
\bigg)\,.
\end{gather}
Notice that the series is convergent because $F,G$ are assumed to be polynomial functionals.
The $*$-algebra $\mathcal{A}_\omega$ obtained by equipping $\mathcal{P}_{\textrm{mloc}}$ with the $\star$-product \eqref{Equation: definition of star-product} and the $*$-involution given by complex conjugation is called the \textit{algebra of $\omega$-renormalized quantum observables}.
Different choices of $\omega$ lead to $*$-isomorphic algebras: This is a consequence of the fact that, if $\omega,\omega'$ are Hadamard distributions, then $\omega-\omega'\in C^\infty(M^2)$ \cite{Ra96,RaVe96}.

States on $\mathcal{A}_\omega$ are defined as linear, positive and normalized functionals $\Omega\colon\mathcal{A}_\omega\to\mathbb{C}$.
Among all the possible choices, we will mainly consider the one obtained considering the evaluation functional $\Omega(F):=F(0)$.
This defines a so-called quasi-free state \cite[Chap.5]{BrDaFrYn15}, namely a state entirely determined by the distribution $C_{\textrm{c}}^\infty(M)^2\ni(f_1,f_2)\mapsto\Omega(F_{f_1}\star_\omega F_{f_2})$, where $F_{f_k}(\phi)=\int_M f_k\phi$.
The latter distribution is called the two-point function associated to $\Omega$ and coincides with $\omega$.
In general, the two-point function $\omega$ of a Poincaré invariant Hadamard state $\Omega$ can be Fourier expanded
as follows: for all $f,g\in C^\infty_{\textrm{c}}(M)$
\begin{gather}\label{Equation: general momentum expansion of Hadamard distribution}
\omega(f,g)=
\int_{\mathbb{R}^3}\frac{\textrm{d}k}{2\epsilon}
\sum_{\pm}c_\pm(k)\widehat{f}(\pm\epsilon,k)\widehat{g}(\mp\epsilon,-k)\,,\qquad
\epsilon=\epsilon(k):=\sqrt{|k|^2+m^2}\,.
\end{gather}
The functions $c_\pm$ identify completely the state $\Omega$.
Actually $\omega$ is an Hadamard distribution if and only if $c_++c_-\geq 0$, $c_+-c_-=1$ and $c_-$ is smooth and rapidly decreasing.

The algebra $\mathcal{A}_\omega$ is an algebra of \textit{off-shell} functionals, namely functionals which are not constrained by any dynamical requirement.
The \textit{on-shell} algebra of quantum observables $\mathcal{A}_{\omega,\textrm{on}}$ is identified with the quotient of $\mathcal{A}_\omega/\mathcal{I}_\omega$ with respect to the $*$-ideal $\mathcal{I}_\omega$ which contains ``dynamically trivial'' functionals.
For the case of an Hadamard distribution $\omega$ which is a weak bisolution of $\square+m^2$ the ideal $\mathcal{I}_\omega$ consists of functionals vanishing on solutions of the Klein-Gordon equation $\square\phi+m^2\phi=0$.

The on-shell algebra enjoys the remarkable property of the time-slice axiom \cite{ChFr09}, which is described as follows.
Let $\mathcal{O}\subseteq M$ be a region of $M$ such that $J(\mathcal{O}):=J^\uparrow(\mathcal{O})\cup J^\downarrow(\mathcal{O})=M$, where $J^\uparrow(\mathcal{O})$ (resp. $J^\downarrow(\mathcal{O})$) denotes the causal future (resp. past) of $\mathcal{O}$ \cite{Ba15,BaGi12}.
Let $\mathcal{A}_{\omega,\textrm{on}}(\mathcal{O})$ be the on-shell algebra generated by $F\in\mathcal{P}_{\textrm{loc}}$ with $\textrm{spt}(F)\subseteq\mathcal{O}$.
This algebra is clearly embedded in the whole algebra $\mathcal{A}_{\omega,\textrm{on}}$: the time-slice axiom ensures that this embedding is in fact a $*$-isomorphism.
In the following, we will mostly deal with the off-shell algebra.

\subsection{Interacting theory}
Interactions for a real scalar Klein-Gordon field are non-linear corrections to the linear operator $\square+m^2$ which are described by a self-adjoint element of the algebra $V\in\mathcal{A}_\omega$ \cite[Chap.2]{BrDaFrYn15}.
This amounts to assume that the dynamics of the interacting field is ruled by the operator $\square+m^2+\lambda V^{(1)}[\cdot]$, where $\lambda$ is the coupling of the interaction.
Notice that $V$ has compact support so that a perturbative approach is justified: The interacting observables are then expanded in formal power series of $\lambda$ -- which is regarded as a formal parameter -- leading to elements in $\mathcal{P}_{\textrm{mloc}}[[\lambda]]$.

Once this step has been accomplished, it remains to discuss the so-called adiabatic limit, where a suitable limit $\textrm{spt}(V)\to M$ is considered.
In the algebraic setting, this is a two-steps procedure.
On the one hand, the adiabatic limit can be performed at the level of algebras, the so-called algebraic adiabatic limit, leading a $*$-algebra $\mathcal{A}_{V,\textrm{ad}}$.
On the other hand, the adiabatic limit $\textrm{spt}(V)\to 1$ can also be considered on family of functionals $\{\Omega_f\}_f$ such that, for each test function $f$, $\Omega_f$ defines a state for the interacting algebra $\mathcal{A}_V$ with $\textrm{spt}(V)=\textrm{spt}(f)$ -- see for example the family of states identified by \eqref{Equation: definition of interacting thermal state}.
The distributional limit $f\to 1$ is defined in an appropriately sense -- \textit{cf.} Section \ref{Section: Main result} -- and its analysis is ultimately a case-by-case study.
The purpose of this paper is to show the convergence of a particular sequence of states for the algebra obtained with a quadratic interaction $V$.

In the following we briefly sketch the construction of the algebras of interacting observables associated with a perturbation $V\{f\}$, where the notation stresses the dependence of $V$ on the cut-off $f\in C^\infty_{\textrm{c}}(M)$, that is $\textrm{spt}(V\{f\})=\textrm{spt}(f)$.
We will not discuss the construction in full details, referring instead to the vast literature on the topic \cite{BaFr09,BrDaFrYn15,BrFr00,ChFr09,DuFr04,EpGl73,FrRe14,HoWa01,HoWa02,HoWa05}.
In this section we assume that a choice for an Hadamard distribution $\omega$ has been made and denote with $\mathcal{A}:=\mathcal{A}_\omega$ the corresponding algebra.

\subsubsection{Quantum M\o ller operator}
Following \cite[Chap.2]{BrDaFrYn15}, the $*$-algebra of interacting observables of $\mathcal{A}_{V\{f\}}$ is introduced as a $*$-subalgebra of $\mathcal{A}[[\lambda]]$.
This $*$-subalgebra is defined through the so-called quantum M\o ller operator $\mathsf{R}_{\lambda V\{f\}}$  \cite{BaFr09}, which can be defined as a map $\mathsf{R}_{\lambda V\{f\}}\colon\mathcal{P}_{\textrm{loc}}\to\mathcal{A}[[\lambda]]$ by the well-known Bogoliubov formula -- see equation \eqref{Equation: quantum Moller operator}.
The definition of this latter maps requires the introduction of the time-ordered product $\cdot_T$ \cite[Chap.2]{BrDaFrYn15} \cite{HoWa02}.
This is an associative and commutative product on the $*$-subalgebra $\mathcal{P}_{\textrm{mreg}}\subset\mathcal{P}_{\textrm{mloc}}$ made of polynomial functionals with smooth functional derivatives of all orders.
The time-ordered product can be extended to the whole $\mathcal{P}_{\textrm{mloc}}$ with a non-unique extension procedure \cite{EpGl73}, where the ambiguities in the extension are controlled by the so-called renormalization freedoms \cite{HoWa01,HoWa02,HoWa05}.
Once an extension of the time-ordered product has been identified, the quantum M\o ller operator is defined through the Bogoliubov formula
\begin{gather}\label{Equation: quantum Moller operator}
\mathsf{R}_{\lambda V\{f\}}(F):=
\exp_T\big[
i  \lambda V
\big]^{-1}\star_\omega\big(
\exp_T\big[
i\lambda V\big]\cdot_T F
\big)\in\mathcal{A}[[\lambda]]\,,
\end{gather}
where $F\in\mathcal{P}_{\textrm{loc}}$ and $\exp_T$ denotes the exponential computed with the time-ordered product while $\exp_T\big[iV\big]^{-1}$ is the inverse of $\exp_T\big[iV\big]$ with respect to $\star_\omega$.
For the sake of simplicity we just summarize the construction as a definition:
\begin{Definition}
Let $V\{f\}\in\mathcal{P}_{\textrm{loc}}$.
The $*$-algebra of interacting observables for the real scalar Klein-Gordon theory associated with the perturbation $V\{f\}$ is the $*$-subalgebra $\mathcal{A}_{V\{f\}}\subset\mathcal{A}[[\lambda]]$ generated by $\mathsf{R}_{\lambda V\{f\}}(\mathcal{P}_{\textrm{loc}})$.
\end{Definition}
The algebraic adiabatic limit is related to the following properties of the quantum M\o ller operator \cite[Chap.1-2]{BrDaFrYn15}.
Let $f_1,f_2\in C^\infty_{\textrm{c}}(M)$ and $F\in\mathcal{P}_{\textrm{loc}}$, then 
\begin{align}\label{Equation: causal property of Quantum Moller Operator}
	\mathsf{R}_{\lambda V\{f_1\}}(F)=F\,,\qquad
	\textrm{ if }
	J^\downarrow(\textrm{spt}(F))\cap J^\uparrow(\textrm{spt}(V\{f_1\}))=\emptyset\,.
\end{align}
Similarly, if $J^\downarrow(\textrm{spt}(V\{f_1-f_2\}))\cap J^\uparrow(\textrm{spt}(F))=\emptyset$ then there exists a formal unitary $U_{f_1,f_2}\in\mathcal{A[[\lambda]]}$ such that
\begin{align}\label{Equation: causal property of Quantum Moller Operator for adiabatic limit}
	\mathsf{R}_{\lambda V\{f_1\}}=
	U_{f_1,f_2}^{-1}\star
	\mathsf{R}_{\lambda V\{f_2\}}(F)\star
	U_{f_1,f_2}\,.
\end{align}

Out of properties (\ref{Equation: causal property of Quantum Moller Operator}-\ref{Equation: causal property of Quantum Moller Operator for adiabatic limit}) the algebraic adiabatic limit can be performed, leading to a $*$-algebra $\mathcal{A}_{V,\textrm{ad}}$, independent from the cut-off of $V$ \cite[Chap.2]{BrDaFrYn15}.
Actually one considers a net of algebras $\mathcal{O}\mapsto\mathcal{A}_{V\{f\}}(\mathcal{O})$ where $\mathcal{O}$ is any double cone of $M$, that is, there exists $x,y\in M$ such that $\mathcal{O}=J^\uparrow\{x\}\cap J^\downarrow\{y\}$.
For each of these algebra one considers the cut-off $f$ to be in the class $1_{\mathcal{O}}$ of functions $g\in C^\infty_c(M)$ such that $g|_{\mathcal{O}}=1$.
Thanks to property \eqref{Equation: causal property of Quantum Moller Operator for adiabatic limit}, for all $f,g\in 1_{\mathcal{O}}$ the algebras $\mathcal{A}_{V\{f\}}(\mathcal{O})$ and $\mathcal{A}_{V\{g\}}(\mathcal{O})$ are unitary equivalent.
This allows to identify, for each double cone $\mathcal{O}$, the algebra $\mathcal{A}_{V,\textrm{ad}}(\mathcal{O})$ as a direct limit, leading to a net of $*$-algebras in the sense of Haag and Kastler \cite{HaKa64}.
The global algebra $\mathcal{A}_{V,\textrm{ad}}$ can then be identified with the direct limit of this net.

Finally, the interacting $*$-algebra $\mathcal{A}_{V\{f\}}\subseteq\mathcal{A}[[\lambda]]$ can be projected on its on-shell version $\mathcal{A}_{V\{f\},\textrm{on}}:=\mathcal{A}_{V\{f\}}/\mathcal{I}_{V\{f\}}$ where $\mathcal{I}_{V\{f\}}:=\mathcal{A}_{V\{f\}}\cap\mathcal{I}_\omega$.
As for $\mathcal{A}$, the time-slice axiom holds true for $\mathcal{A}_{V\{f\},\textrm{on}}$ as well as for $\mathcal{A}_{V,\textrm{ad},\textrm{on}}$ \cite{ChFr09}.
Once again, the whole construction can be shown to be covariant in the sense of a generally covariant local theory introduced in \cite{BrFrVe03,Kay92} 

In what follows, we will exploit the time-slice axiom and the covariance of the construction.
Indeed, we will focus on the off-shell algebra $\mathcal{A}_{V\{f\}}(J^\uparrow(\Sigma))$, with $\Sigma$ being a Cauchy surface for $M$ \cite{BaGi12,BeSa05,BeSa06}.
If not stated otherwise, in the following we will leave the $\Sigma$-dependence of $\mathcal{A}_{V\{f\}}(J^\uparrow(\Sigma))$ implicit.
In particular, exploiting an arbitrary but fixed inertial frame for which $\Sigma=t^{-1}\{0\}$, we will choose the cut-off $f$ as a product $h\chi$, where $h\in C^\infty_{\textrm{c}}(\mathbb{R}^3)$ and $\chi\in C^\infty_{\textrm{c}}(\mathbb{R})$ with $\textrm{spt}(\chi)\subseteq(-1,+\infty)$.
Moreover, due to property \eqref{Equation: causal property of Quantum Moller Operator} it is not restrictive to assume $\chi\in C^\infty(\mathbb{R})$ be such that $\chi(t)=1$ for $t\geq 0$.

\subsubsection{Interacting thermal states}
In this section we summarize the construction proposed in \cite{FrLi14}.
The latter aims to define an interacting thermal state out of an arbitrary chosen thermal state for the free theory.
This construction is inspired by analogy to the one proposed in \cite{Ar73} in the framework of $C^*$-algebras.

Thermal equilibrium states are identified by the so-called Kubo-Martin-Schwinger (KMS) condition \cite{BrKiRo78,BrRo97a,BrRo97b,HaHuWi67,HaKaPo74,Ro73,SaVe00}.
In the algebraic approach the latter requires the identification of a one-parameter group of automorphism on the algebra of interest, which is interpreted as the group of time translations.
In the case of the free algebra $\mathcal{A}$ this can be defined as follows.
For all $t\in\mathbb{R}$ and $\phi\in C^\infty(M)$ let $\phi_t$ be the time translation of $\phi$ by $t$ (we implicitly fixed an inertial frame).
Then, for all $F\in\mathcal{A}$ the time translation is defined as
\begin{gather}\label{Equation: free one-parameter group}
F\mapsto \tau_t(F):=F_t\,,\quad
F_t(\phi):=
F(\phi_t)\,.
\end{gather}
As for the interacting algebra the one-parameter group is defined on each generator of $\mathcal{A}_{V\{f\}}$ as \cite{FrLi14}
\begin{gather}\label{Equation: interacting one-parameter group}
\mathsf{R}_{\lambda V\{f\}}(F)\mapsto
\tau_{V\{f\},t}\big(
\mathsf{R}_{\lambda V\{f\}}(F)
\big):=
\mathsf{R}_{\lambda V\{f\}}(F_t)\,.
\end{gather}
Once a one-parameter group of $*$-automorphism has been fixed one may introduce KMS states as follows \cite{BrRo97a,BrRo97b,SaVe00}.
\begin{Definition}\label{Definition: KMS state}
Let $\mathsf{A}$ be a topological $*$-algebra and let $\alpha\in\hom(\mathbb{R},\textrm{Aut}(\mathsf{A}))$ be a one-parameter group of $*$-automorphism of $\mathsf{A}$.
A state $\Omega$ over $\mathsf{A}$ is called a $(\beta,\alpha)$-KMS state at inverse temperature $\beta>0$ if, for all $a,b\in\mathsf{A}$, the function $t\mapsto\Omega(a\alpha_t(b))$ admits an analytic continuation -- denoted with $\Omega(a\alpha_z(b))$ -- in the complex strip $\mathcal{S}_\beta:=\{z\in\mathbb{C}|\; 0<\Im z<\beta \}$ which is continuous on the closure $\overline{\mathcal{S}_\beta}$ and such that
\begin{gather}\label{Equation: KMS condition}
\Omega(a\alpha_z(b))|_{z=i\beta}=
\Omega(ba)\,.
\end{gather}
\end{Definition}
In the case of the free algebra $\mathcal{A}$, for all $\beta>0$ there is a unique KMS state $\Omega_\beta$ which is a quasi-free state whose two-point function is given by, \textit{cf.} expression \eqref{Equation: general momentum expansion of Hadamard distribution},
\begin{align}\label{Equation: free thermal state}
	\omega_\beta(f,g):=
	\int_{\mathbb{R}^3}
	\frac{\textrm{d}k}{2\epsilon}
	\sum_{\pm}b_\pm(\beta,\epsilon)\widehat{f}(\pm\epsilon,k)\widehat{g}(\mp\epsilon,-k)
	\,,\qquad
	b_\pm(\beta,\epsilon):=\frac{\mp 1}{e^{\mp\beta\epsilon}-1}\,.
\end{align}
The identity $b_-(\beta,\epsilon)=e^{-\beta\epsilon}b_+(\beta,\epsilon)$ ensures the KMS condition \eqref{Equation: KMS condition} as well as the Hadamard property.

Let now $V\{h\chi\}\in\mathcal{P}_{\textrm{loc}}$.
Following a previous construction in the context of $C^*$-algebras \cite{Ar73}, in \cite{FrLi14} the causality properties (\ref{Equation: causal property of Quantum Moller Operator}-\ref{Equation: causal property of Quantum Moller Operator for adiabatic limit}) were exploited to built an intertwiner between the free time evolution $\tau$ and the interacting time evolution $\tau_{V\{h\chi\}}$.
Actually, for all $F\in\mathcal{A}_{V\{h\chi\}}$ there exists a unitary cocycle $U_{V\{h\chi\}}(t)\in\mathcal{A}[[\lambda]]$ such that
\begin{align}\label{Equation: intertwining property of the cocycle}
	\tau_{V\{h\chi\},t}
	\big[\mathsf{R}_{\lambda V\{h\chi\}}(F)\big]=
	U_{V\{h\chi\}}(t)^{-1}\star_\beta
	\tau_t
	\big[\mathsf{R}_{\lambda V\{h\chi \}}(F)\big]
	\star_\beta
	U_{V\{h\chi \}}(t)\,.
\end{align}
Here we have implicitly identified the free algebra $\mathcal{A}$ with the $\omega_\beta$-renormalized algebra $\mathcal{A}_{\omega_\beta}$ and $\star_\beta:=\star_{\omega_\beta}$ is a short notation.
The cocycle $U_{V\{h\chi \}}(t)$ satisfies the cocycle condition
\begin{align}\label{Equation: cocycle condition}
	U_{V\{h\chi\}}(t+s)=
	U_{V\{h\chi\}}(t)\star_\beta\tau_t\big[U_{V\{h\chi\}}(s)\big]\,,
\end{align}
which has a cohomological interpretation \cite{BuRo76}.
Property \eqref{Equation: cocycle condition} implies that the state
\begin{gather}\label{Equation: definition of interacting thermal state}
	\Omega_{\beta,\lambda V\{h\chi\}}(A):=
	\frac{\Omega_\beta
	\big(
	A\star_\beta U_{V\{h\chi \}}(t)
	\big)}{\Omega_\beta\big(U_{V\{h\chi\}}(t)\big)}\Bigg|_{t=i\beta}
	\qquad\forall A\in\mathcal{A}_{V\{h\chi\}}\,,
\end{gather}
is a well-defined $(\beta,\tau_{V\{h\chi\}})$-KMS state for the interacting algebra $\mathcal{A}_{V\{h\chi\}}$ \cite{FrLi14}.
Notice that the evalutation at $t=i\beta$ is justified at each order in $\lambda$ by the analytic properties of $\Omega_\beta$.

The state $\Omega_{\beta,\lambda V\{h\chi\}}$ enjoys the following expansion \cite[Prop. 3]{FrLi14}
\begin{align}\label{Equation: interacting thermal state}
	\Omega_{\beta,\lambda V\{h\chi\}}(A)=
	\Omega_\beta(A)+
	\sum_{n\geq 1}(-1)^n\int_{\beta S_n}\textrm{d}U\,\Omega_\beta^c
	\Bigg[
	A\otimes\bigotimes_{\ell=1}^n K_{iu_\ell}
	\Bigg]\,,\qquad
	\forall A\in\mathcal{A}_{V\{h\chi\}}\,.
\end{align}
Here, $S_n:=\{U:=(u_1,\ldots,u_n)\in\mathbb{R}^n|\quad 0\leq u_1\leq\ldots\leq u_n\leq 1\}$ is the canonical $n$-dimensional simplex while $K:=\frac{\textrm{d}}{i\textrm{d}t}U_{V\{h\chi\}}(t)\big|_{t=0}=\mathsf{R}_{\lambda V\{h\chi\}}(\lambda V\{h\dot{\chi}\})$, the dot being time derivative.
Moreover, $\Omega_\beta^c$ denotes the connected part of $\Omega_\beta$
which is defined by
\begin{align}\label{Equation: definition of connected functions}
	\Omega^\beta(A_1\star\ldots\star A_n)=
	\sum_{P\in\mathsf{P}\{1,\ldots,n\}}\prod_{I\in P}\Omega_\beta^{\textrm{c}}\bigg(\bigotimes_{\ell\in I}A_\ell\bigg)\,,
	\qquad\forall A_1,\ldots,A_n\in\mathcal{A}\,,\forall n\in\mathbb{Z}_+\,,
\end{align}
together with the condition $\Omega_\beta^{\textrm{c}}(1_{\mathcal{A}})=0$ -- here $\mathsf{P}\{1,\ldots,n\}$ denotes the set of partition of $\{1,\ldots,n\}$ in non-empty subsets.

In \cite{FrLi14} the dependence of the state $\Omega_{\beta,\lambda V\{h\chi\}}$ on the cut-off present in $V$ was studied.
In the massive case, the clustering properties of the state $\Omega_\beta$ guarantees that the limit $h\to 1$ can be performed in the sense of van Hove \cite[Def. 2]{FrLi14}, leading to a KMS state $\Omega_{\beta,\lambda V\{\chi\}}$.
In \cite{DrHaPi16} similar conclusions were drawn also in the massless case, where the lack of clustering properties for $\Omega_\beta$ can be treated by exploiting the so-called principle of perturbative agreement (PPA) \cite{HoWa05}.
In \cite{DrFaPi17}, the long time behaviour of the state $\Omega_{\beta,\lambda V\{\chi\}}$ has been investigated, leading to one of the first example of non-equilibrium steady state (NESS) \cite{Ru00} in the context of Quantum Field Theory -- see also \cite{Hack-Verch-2018}.
In particular this state is defined as the weak limit
\begin{gather}\label{Equation: definition of the NESS}
\Omega_{\textrm{\textsc{ness}}}(A):=
\lim_{t\to\infty}\frac 1t
\int_0^t\textrm{d}s\,\Omega_{\beta,\lambda V\{\chi\}}[\tau_s(A)]\,,
\qquad\forall A\in\mathcal{A}_{V\{\chi\}}\,.
\end{gather}
The thermodynamical properties of $\Omega_{\textrm{\textsc{ness}}}$ have been discussed in \cite{Drago-Faldino-Pinamonti-18}.

\section{Main result}\label{Section: Main result}
The goal of this paper is to study the adiabatic limit of the interacting thermal state $\Omega_{\beta,\lambda Q\{h\chi\}}$ constructed in \cite{FrLi14} in the case of a perturbation given by a quadratic interaction $Q$:
\begin{gather}\label{Equation: quadratic perturbation}
Q\{h\chi\}(\phi):=\frac{m_0^2}{2}\int_M h\chi\phi^2\,,
\end{gather}
where $m_0^2>0$ has the dimension of a squared mass.
The precise definition of the adiabatic limit is the following: After performing the van Hove limit $h\to 1$ of \eqref{Equation: interacting thermal state} we obtain a state $\Omega_{\beta,\lambda Q\{\chi\}}$, which  depends on $\chi$ \cite{FrLi14}.
We then consider $\chi\in C^\infty(\mathbb{R})$ be such that
\begin{align}\label{Equation: chi-properties}
	\textrm{spt}(\chi)\subseteq(-1,+\infty)\,,\qquad\chi(t)=1\,\quad\textrm{ for } t\geq 0\,,
\end{align}
and we set $\chi_\mu(t):=\chi(t/\mu)$ for all $\mu>0$.
We denote with $Q\{\chi_\mu\}$ and $\Omega_{\beta,\lambda Q\{\chi_\mu\}}$ the quadratic perturbation \eqref{Equation: quadratic perturbation} and the corresponding state \eqref{Equation: interacting thermal state} where $\chi$ has been substituted with $\chi_\mu$.
The adiabatic limit of $\Omega_{\beta,\lambda Q\{h\chi\}}$ is, per definition, the weak limit of the sequence of states $\Omega_{\beta,\lambda Q\{\chi_\mu\}}$ as $\mu\to+\infty$.
Notice that this prescription for the adiabatic limit consists in removing the space cut-off $h$ -- which is the ``thermodynamical limit" for the interaction $\lambda Q\{h\chi\}$ -- and then in addressing the limit for the time cut-off $\chi$.
Exchanging the limits $h\to 1,\chi\to 1$ described above would probably lead to a rather trivial result because for fixed $h$ observables $F$ with $\textrm{spt}(F)\cap J(\textrm{spt}(h))=\emptyset$ remain unaffected by the action of the M\o ller operator $\mathsf{R}_{\lambda Q\{h\chi\}}$ -- \textit{cf.} equation \eqref{Equation: causal property of Quantum Moller Operator}.

Thanks to the structural assumption on $Q\{\chi_\mu\}$ -- see in particular equation \eqref{Equation: PPA} -- we will be able to compute the limit $\mu\to+\infty$ of each term in the series \eqref{Equation: interacting thermal state} for $\Omega_{\beta,\lambda Q\{\chi_\mu\}}$.
Moreover, we will be also able to give a closed form for the series itself: Then the resulting state is compared to the KMS state on the algebra of the Klein-Gordon theory with mass $m^2+\lambda m_0^2$.
For the convenience of the reader we state here the main result:
\begin{theorem}\label{Theorem: main theorem}
For a quadratic perturbation $Q\{\chi_\mu\}$ as in \eqref{Equation: quadratic perturbation}, the state $\Omega_{\textrm{\textsc{ad}}}$ defined as the weak limit of the sequence $\Omega_{\beta,\lambda Q\{\chi_\mu\}}$ for $\mu\to+\infty$ is the quasi-free state whose two-point function reads
\begin{gather}\label{Equation: two point function of the NESS state}
\omega_{\textrm{\textsc{ad}}}(f,g)=
\int_{\mathbb{R}^3}
\frac{\textrm{d}k}{2\epsilon_\lambda}
\sum_{\pm}b_\pm(\beta,\epsilon_\lambda)\widehat{f}(\pm\epsilon_{\lambda},k)\widehat{g}(\mp\epsilon_{\lambda},-k)\,,\qquad
\epsilon_\lambda=\epsilon_\lambda(k):=\sqrt{|k|^2+m^2+\lambda m_0^2}\,,
\end{gather}
In other words, in the adiabatic limit, $\Omega_{\beta,\lambda Q\{\chi_\mu\}}$ converges to the KMS state for the Klein-Gordon theory with mass $m^2+\lambda m_0^2$.
\end{theorem}
\begin{remark}
\textit{(i)}
We stress that we do not make any claims about the case of a tachyonic (imaginary) mass $m^2+\lambda m_0^2<0$.
\textit{(ii)}
Notice that the general form of a purely quadratic local functional would contain first derivatives of the field, \textit{i.e.}, terms proportional to $\partial_a\phi\partial_b\phi,\phi\partial_a\phi$.
While the latter term would not be a great deal and may model the presence of an external heat flux (described by an interaction $\sim\phi Q^a\partial_a\phi$), the former one would spoil some of the result of \cite{DrHaPi16} which we will need in the following sections.
\end{remark}

\subsection{Preliminary observations}
\subsubsection{Quantum M\o ller operator}\label{Section: Quantum Moller operator}
In this section we describe how the simple structure of the interaction potential $Q\{\chi_\mu\}$ given in \eqref{Equation: quadratic perturbation} allows to simplify the expression involving the quantum M\o ller operator.

Indeed, with reference to \cite{DrHaPi16}, we recall that for any quadratic local functional $Q\{\chi_\mu\}$ it holds
\begin{gather}\label{Equation: PPA}
	\mathsf{R}_{\lambda Q\{\chi_\mu\}}=
	\mathsf{R}^{\textrm{cl}}_{\lambda Q\{\chi_\mu\}}\circ\gamma_{\lambda Q\{\chi_\mu\}}\,,
\end{gather}
where $\mathsf{R}^{\textrm{cl}}_{\lambda Q\{\chi_\mu\}}$ is the classical M\o ller operator \cite{Pe52} while $\gamma_{\lambda Q\{\chi_\mu\}}$ is a contraction map between local functionals -- see \cite{DrHaPi16} for details.
The classical M\o ller operator $\mathsf{R}^{\textrm{cl}}_{\lambda Q\{\chi_\mu\}}\colon\mathcal{P}_{\textrm{loc}}\to\mathcal{A}[[\lambda]]$ can be thought as the classical limit of $\mathsf{R}_{\lambda Q\{\chi_\mu\}}$ \cite{DaDr16,DrHaPi16,DuFr01bis,DuFr04}.
Actually, it is an \textit{exact}, \textit{i.e.} non-perturbative, $*$-isomorphism $\mathsf{R}^{\textrm{cl}}_{\lambda Q\{\chi_\mu\}}\colon\mathcal{A}_{\lambda Q\{\chi_\mu\}}\to\mathcal{A}$ between the algebra $\mathcal{A}_{\lambda Q\{\chi_\mu\}}$ of quantum observables associated to the free Klein-Gordon field whose dynamics is ruled by the operator $\square+m^2+\lambda m_0^2\chi_\mu$ and the algebra $\mathcal{A}$.
Its pull-back action on states has been studied in \cite{DaDr16,Drago-Murro-2017} and will be exploited in the following -- \textit{cf.} equation \eqref{Equation: two point function of KMS under the pull-back of classical Moller}.

The main feature of $\gamma_{\lambda Q\{\chi_\mu\}}$ is that it does not increase the number of fields present in each observable.
Hence, when applied on linear or quadratic functionals, $\gamma_{\lambda Q\{\chi_\mu\}}$ is the identity up to constant, actually
\begin{gather*}
	\gamma_{\lambda Q\{\chi_\mu\}}(F)= F,\qquad
	\gamma_{\lambda Q\{\chi_\mu\}}(Q') = Q'+c\,,
\end{gather*}
where $F(\phi):=\int_Mf\phi$ and $Q'$ is any quadratic functional.
The conclusion is that, on linear and quadratic functionals, the quantum M\o ller operator $\mathsf{R}_{\lambda Q\{\chi_\mu\}}$ and the classical M\o ller operator $\mathsf{R}^{\textrm{cl}}_{\lambda Q\{\chi_\mu\}}$ coincide up to constant.
The latter will play no r\^ole in the subsequent discussion due to the presence of the connected part $\Omega_\beta^{\textrm{c}}$ of $\Omega_\beta$.

This observation allows to rewrite the connected part $\Omega_\beta^{\textrm{c}}$ of $\Omega_\beta$ appearing in \eqref{Equation: interacting thermal state} as
\begin{align*}
	\Omega_\beta^{\textrm{c}}\bigg[
	\mathsf{R}_{\lambda Q\{\chi_\mu\}}(F_1)\star_\beta\dots&\star_\beta\mathsf{R}_{\lambda Q\{\chi_\mu\}}(F_n)
	\otimes\bigotimes_{\ell=1}^n\big[\mathsf{R}_{\lambda Q\{\chi_\mu\}}Q\{\dot{\chi}_\mu\}\big]_{i u_\ell}
	\bigg]\\ &=
	\Omega_\beta^{\textrm{c}}\bigg[
	\mathsf{R}^{\textrm{cl}}_{\lambda Q\{\chi_\mu\}}(F_1)\star_\beta\dots\star_\beta\mathsf{R}^{\textrm{cl}}_{\lambda Q\{\chi_\mu\}}(F_n)
	\otimes\bigotimes_{\ell=1}^n\big[\mathsf{R}^{\textrm{cl}}_{\lambda Q\{\chi_\mu\}}Q\{\dot{\chi}_\mu\}\big]_{i u_\ell}
	\bigg]\\ &=
	\big[
	\Omega_\beta\circ\mathsf{R}^{\textrm{cl}}_{\lambda Q\{\chi_\mu\}}
	\big]^{\textrm{c}}
	\bigg[
	F_1\star_{\lambda Q\{\chi_\mu\}}\dots\star_{\lambda Q\{\chi_\mu\}} F_n
	\otimes\bigotimes_{\ell=1}^nQ\{\dot{\chi}_\mu\}_{i u_\ell}
	\bigg]\,,
\end{align*}
where
we exploited the fact $\Omega_\beta^{\textrm{c}}(1_\mathcal{A})=0$.
The $\star$-product $\star_{\lambda Q\{\chi_\mu\}}$ is the one induced by the $*$-isomorphism $\mathsf{R}^{\textrm{cl}}_{\lambda Q\{\chi_\mu\}}$, \textit{i.e.} it is induced by $\omega_\beta\circ\mathsf{R}^{\textrm{cl}}_{\lambda Q\{\chi_\mu\}}$.
This computation shows that the $\chi_\mu$-dependent part of \eqref{Equation: interacting thermal state} is either in the appearance of $Q\{\dot{\chi}_\mu\}$-terms in the connected function or in the presence of the pull-back state $\Omega_\beta\circ\mathsf{R}^{\textrm{cl}}_{\lambda Q\{\chi_\mu\}}$.

In the following we will compute the limit $\mu\to+\infty$ of each term of the series \eqref{Equation: interacting thermal state}, exploiting the exact knowledge provided by $\mathsf{R}^{\textrm{cl}}_{\lambda Q\{\chi_\mu\}}$.
In a sense, equation \eqref{Equation: PPA} allows to switch to an effective description where the classical part of the perturbative series in $\lambda$ has already been summed, leaving untouched the pure quantum contributions.
This simplification is expected to hold for a more general potential $V$, though it would probably appear as $\mathsf{R}_{\lambda V\{f\}}=\mathsf{R}^{\textrm{cl}}_{\lambda V_{\textrm{eff}}\{f\}}\circ\gamma_{\lambda V\{f\}}$, with $V_{\textrm{eff}}\{f\}$ being an effective potential, perturbatively built out of $V\{f\}$ \cite{BrDu08,We96}.

As explained in \cite{DaDr16,DrHaPi16} the state $\Omega_\beta\circ\mathsf{R}^{\textrm{cl}}_{\lambda Q\{\chi_\mu\}}$ is a quasi-free state whose two-point function $\omega_{\lambda V\{\chi_\mu\}}$ is given by
\begin{gather}\label{Equation: two point function of KMS under the pull-back of classical Moller}
\omega_{\lambda V\{\chi_\mu\}}(f,g)=
\int_{\mathbb{R}}\textrm{d}t
\int_{\mathbb{R}}\textrm{d}t'
\int_{\mathbb{R}^3}\textrm{d}k\,
\widetilde{f}(t,k)
\widetilde{g}(t',k)\bigg[
b_+(\beta,\epsilon)T_{k,\mu}(t)\overline{T_{k,\mu}(t')}+
b_-(\beta,\epsilon)\overline{T_{k,\mu}(t)}T_{k,\mu}(t')
\bigg]\,,
\end{gather}
where $f,g\in C^\infty_{\textrm{c}}(\mathbb{R}^4)$ and $\widetilde{f},\widetilde{g}$ denotes the Fourier transform of $f,g$ in three momentum $k\in\mathbb{R}^3$.
The modes $T_{k,\mu}(t)$ are solutions of the following differential equation
\begin{align}\label{Equation: interpolating modes}
	\ddot{T}_{k,\mu}(t)+\epsilon_\mu(k,t)^2T_{k,\mu}(t)=0\,,\qquad
	T_{k,\mu}(t)=\frac{e^{-i\epsilon t}}{\sqrt{2\epsilon}}
	\qquad\textrm{ for }
	t\notin\textrm{spt}(\chi)\,,
\end{align}
where $\epsilon_\mu(k,\tau):=\sqrt{\epsilon(k)^2+(\epsilon_\lambda(k)^2-\epsilon(k)^2)\chi_\mu(t)}$, subject to the Wronskian condition
\begin{gather}\label{Equation: Wronski condition}
\overline{\dot{T}}_{k,\mu}(t)T_{k,\mu}(t)-\overline{T_{k,\mu}(t)}\dot{T}_{k,\mu}(t)=i\,.
\end{gather}

The limit of this latter state as $\mu\to+\infty$ was investigated in \cite{DaDr16,DrGe16,DrHaPi16}:
\begin{lemma}[\cite{DrGe16}]\label{Lemma: classical adiabatic limit of KMS state}
The limit $\mu\to+\infty$ of the sequence $\Omega_\beta\circ\mathsf{R}^{\textrm{cl}}_{\lambda Q\{\chi_\mu\}}$ exists and defines a quasi-free state $\Omega_{\textrm{\textsc{ad},cl}}$ whose two point function $\omega_{\textrm{\textsc{ad},cl}}$ reads
\begin{gather}\label{Equation: two point function of KMS under classical adiabatic limit}
\omega_{\textrm{\textsc{ad},cl}}(f,g):=
\int_{\mathbb{R}^3}
\frac{\textrm{d}k}{2\epsilon_\lambda}
\sum_{\pm}b_\pm(\beta,\epsilon)\widehat{f}(\pm\epsilon_{\lambda},k)\widehat{g}(\mp\epsilon_{\lambda},-k)\,.
\end{gather}
\end{lemma}
\begin{remark}
As remarked in \cite{DrGe16} the limit for $\mu\to+\infty$ of a thermal state $\Omega_\beta$ under the action of the classical M\o ller operator $\mathsf{R}^{\textrm{cl}}_{\lambda Q\{\chi_\mu\}}$ fails to be the corresponding thermal state for the theory with mass $m^2+\lambda m_0^2$.
The reason can be traced back to the fact that, while the $\epsilon$-modes of the states are correctly changed into the $\epsilon_\lambda$ ones -- \textit{i.e.} $T_{k,\mu}(t)\overline{T_{k,\mu}(t')}\to(2\epsilon_\lambda)^{-1}\exp[-i\epsilon_\lambda(t-t')]$ -- the thermal coefficients $b_\pm(\beta,\epsilon)$ remain untouched.
Theorem \ref{Theorem: main theorem} shows that the terms needed to restore the KMS property are exactly those provided by the perturbative series \eqref{Equation: interacting thermal state}.
\end{remark}
For future convenience we state the following lemma.
\begin{lemma}\label{Lemma: useful limits}
	Let $\chi\in C^\infty(\mathbb{R})$ with the property \eqref{Equation: chi-properties} and set $\chi_\mu(t):=\chi(t/\mu)$ for $\mu>0$.
	For any $k\in\mathbb{R}^3$, let $T_{k,\mu}(t)$ be the modes defined as in \eqref{Equation: interpolating modes}.
	Then
	\begin{gather}\label{Equation: useful limits}
		\lim_{\mu\to+\infty}
		\int_{\mathbb{R}}\textrm{d}t\;
		T_{k,\mu}(t)^2
		\frac{\textrm{d}}{\textrm{d}t}\chi_\mu(t)=0\,,\qquad
		\lim_{\mu\to+\infty}
		\int_{\mathbb{R}}\textrm{d}t\;
		\big|
		T_{k,\mu}(t)
		\big|^2
		\frac{\textrm{d}}{\textrm{d}t}\chi_\mu(t)=
		\frac{1}{\epsilon_\lambda+\epsilon}\,.
	\end{gather}
\end{lemma}
\begin{proof}
	We recall a few facts from \cite[Lemma 5.1]{DaDr16}, \cite[Appendix D]{DrHaPi16}.
	First of all we set, for all $k\in\mathbb{R}^3$,
	\begin{align}\label{Equation: auxiliary modes}
		T_{\textrm{a},k,\mu}(t):=
		\frac{1}{\sqrt{2\epsilon_\mu(k,t)}}
		\exp\bigg[
		-i\int_{t_0}^t\textrm{d}s\;
		\epsilon_\mu(k,s)
		\bigg]\,,\qquad
		\epsilon_\mu(k,t)=
		\sqrt{\epsilon^2(k)+(\epsilon_\lambda^2(k)-\epsilon^2(k))\chi_\mu(t)}\,,
	\end{align}
	where $t_0\notin\textrm{spt}(\chi_\mu)$ is arbitrary but fixed.
	Following \cite{DaDr16,DrHaPi16} one finds that, for all $k\in\mathbb{R}^3$ and $t\in\mathbb{R}$, $T_{\textrm{a},k,\mu}(t)\overline{T_{\textrm{a},k,\mu}(t')}\to (2\epsilon_\lambda(k))^{-1}\exp[-i\epsilon_\lambda(k)(t-t')]$ as $\mu\to+\infty$.
	Notice however that $T_{\textrm{a},k,\mu}(t)$ has no limit for $\mu\to+\infty$.
	Finally, $|T_{k,\mu}-T_{\textrm{a},k,\mu}|\to 0$ as $\mu\to+\infty$ uniformly in $t$, see the proof of \cite[Lemma 5.1]{DaDr16} or \cite[Appendix D]{DrHaPi16}.
	
	We now address the second limit in \eqref{Equation: useful limits}:
	\begin{align*}
		\lim_{\mu\to+\infty}
		\int_{\mathbb{R}}\frac{\textrm{d}t}{\mu}\,
		\big|
		T_{k,\mu}(t)
		\big|^2
		\dot{\chi}\bigg(\frac{t}{\mu}\bigg)=
		\lim_{\mu\to+\infty}
		\int_{\mathbb{R}}\textrm{d}t\,
		\big|
		T_{\textrm{a},k,\mu}(\mu t)
		\big|^2
		\dot{\chi}(t)=
		\lim_{\mu\to+\infty}
		\int_{\mathbb{R}}\textrm{d}t\,
		\frac{\dot{\chi}(t)}{2\epsilon_\mu(k,\mu t)}=
		\frac{1}{\epsilon_\lambda(k)+\epsilon(k)}\,.
	\end{align*}
	The first limit in \eqref{Equation: useful limits} is addressed similarly.
	Indeed, since $t_0\leq 0$, one finds, for all $t\in\mathbb{R}$,
	\begin{align*}
		\mu\int_{t_0/\mu}^t
		\epsilon_\mu(k,\mu s)\textrm{d}s-
		\mu\int_0^t
		\epsilon_\mu(k,\mu s)\textrm{d}s
		\longrightarrow_{\mu\to+\infty}
		t_0\epsilon_\lambda(k)\,.
	\end{align*}
	Therefore, as $\mu\to+\infty$, $T_{\textrm{a},k,\mu}(\mu t)^2\simeq(2\epsilon_\mu(k,\mu t))^{-1}\exp\big(i\mu\psi(t)\big)$, with $\dot{\psi}(t)\neq 0$ and the thesis follows by Riemann's Lemma.
\end{proof}

\subsubsection{$\beta$-expansion of the Bose-Einstein factor}\label{Section: beta-expansion for the Bose-Einstein factor}
For later convenience we provide a formula for the $\beta$-derivatives of the ``thermal coefficients'' $b_\pm(\beta,\epsilon)$, which will be useful in the proof of Theorem \ref{Theorem: main theorem}.
We recall that $b_\pm$ are defined by
\begin{gather}
b_\pm(\beta,\epsilon)=
\frac{\mp 1}{e^{\mp\beta\epsilon}-1}\,.
\end{gather}
In particular, $b_-$ is the Bose-Einstein factor.
Thanks to the relations $ b_\pm(\beta,\epsilon)=e^{\pm\beta\epsilon}b_\mp(\beta,\epsilon)$ and $b_+-b_-=1$ one may compute
\begin{gather}\label{Equation: first beta derivative of Bose Einstein factor}
\partial_\beta b_\pm(\beta,\epsilon)=-\epsilon b_+(\beta,\epsilon)b_-(\beta,\epsilon)\,.
\end{gather}
Iterating relation \eqref{Equation: first beta derivative of Bose Einstein factor} we find
\begin{align}\label{Equation: n-order beta derivative of Bose Einstein factor}
	\partial_\beta^nb_\pm =
	(-\epsilon)^n\sum_{k=1}^{n}
	c_{n,k}b_+^{n+1-k}b_-^k\,,
\end{align}
where the coefficients $(c_{n,k})_{n\geq 1,1\leq k\leq n}$ satisfy the following recursion relations:
\begin{align}\label{Equation: recursive relation for the coefficients of the beta derivative}
	c_{n,k}=\Bigg\lbrace
	\begin{array}{l}
	1
	\qquad\textrm{ for }k\in\{1,n\}\\
	kc_{n-1,k}+(n+1-k)c_{n-1,k-1}
	\qquad\textrm{ for }2\leq k \leq n-1\,.
	\end{array}\,,\qquad
	c_{n,k}=c_{n,n+1-k}\,.
\end{align}
From equation \eqref{Equation: recursive relation for the coefficients of the beta derivative} it follows that the coefficients $c_{n,k}$ coincide with the Eulerian numbers $A(n,k)$ \cite[Thm. 1.7]{Bona-12}, that is, $c_{n,k}$ is the number of $n$-permutations with $k-1$ descents.
We recall that, for a given $n$-permutation $\sigma$ of $\{1,\ldots,n\}$, an index $i\in\{1,\ldots,n-1\}$ is called descent of $\sigma$ if $\sigma_i>\sigma_{i+1}$.
In what follows we denote with $\wp_n$ the set of $n$-permutations and with $\wp_{n,k}\subset\wp_n$ the subset of $n$-permutation with $k-1$ descents -- so that $c_{n,k}=|\wp_{n,k}|$.
See \cite{Fewster-Ford-12,Fewster-Siemssen-14} for similar applications of these structures in the QFT framework.

\subsection{Proof of Theorem \ref{Theorem: main theorem}}\label{Section: Proof of Main theorem}
We are now ready to prove Theorem \ref{Theorem: main theorem}.
\begin{proof}(Thm.\ref{Theorem: main theorem})
	We recall that, $Q\{\chi_\mu\}$ is the quadratic functional \eqref{Equation: quadratic perturbation} with cut-off $\chi_\mu$, where $\chi\in C^\infty(\mathbb{R})$ satisfies \eqref{Equation: chi-properties} while $\chi_\mu(t):=\chi(t/\mu)$.
	
	At first, we focus on the case $A=\mathsf{R}_{\lambda Q\{\chi_\mu\}}(F)\star\mathsf{R}_{\lambda Q\{\chi_\mu\}}(G)$, where $F,G\in\mathcal{A}$ are two linear functionals defined by $F(\phi):=\int_M f\phi$, $G(\phi):=\int_M g\phi$ for $f,g\in C^\infty_{\textrm{c}}(M)$.
	The general case will be outlined at the end of the proof.
	
	By exploiting the results of section \ref{Section: Quantum Moller operator} -- see equation \eqref{Equation: PPA} -- we may write the state \eqref{Equation: interacting thermal state} as
	\begin{align}
		\label{Equation: mu-dependent interacting thermal state}
		\Omega_{\beta,\lambda Q\{\chi_\mu\}}\bigg[
		\mathsf{R}_{\lambda Q\{\chi_\mu\}}(F)\star_\beta\mathsf{R}_{\lambda Q\{\chi_\mu\}}(G)
		\bigg]&=
		\omega_{\lambda Q\{\chi_\mu\}}(f,g)\\
		\nonumber
		&+
		\sum_{n\geq 1}(-1)^n\int_{\beta S_n}\textrm{d}U\,
		\big(\Omega_\beta\circ\mathsf{R}^{\textrm{cl}}_{\lambda Q\{\chi_\mu\}}\big)^c
		\Bigg[
		FG\otimes\bigotimes_{\ell=1}^n (\lambda Q\{\dot{\chi}_\mu\})_{iu_\ell}
		\Bigg]\,.
	\end{align}
	Notice that, by Lemma \ref{Lemma: classical adiabatic limit of KMS state}, the first term on the right-hand side of \eqref{Equation: mu-dependent interacting thermal state} tends to $\omega_{\textrm{\textsc{ad},cl}}(f,g)$.
	In what follows, we will consider each term of the series appearing in \eqref{Equation: mu-dependent interacting thermal state} and compute its limit as $\mu\to+\infty$.
	We will then be able to sum the series.
	Notice that the $n$-th term in the series \eqref{Equation: mu-dependent interacting thermal state} is \textit{not} the $n$-th order in perturbation theory of $\Omega_{\beta,\lambda Q\{\chi_\mu\}}$.
	Indeed the state $\Omega_\beta\circ\mathsf{R}^{\textrm{cl}}_{\lambda Q\{\chi_\mu\}}$ depends on the formal parameter $\lambda$.
	The main advantage coming from the results recalled in section \ref{Section: Quantum Moller operator} is that $\Omega_\beta\circ\mathsf{R}^{\textrm{cl}}_{\lambda Q\{\chi_\mu\}}$ can be considered as an exact -- \textit{i.e.} non perturbative -- quantity, and what has to be computed is only the $n$-th order term of the series in \eqref{Equation: mu-dependent interacting thermal state}.
	In other words, we are exploiting a partial summation, in which the contribution in $\lambda$ coming from $\mathsf{R}^{\textrm{cl}}_{\lambda Q\{\chi_\mu\}}$ are recollected.
	
	With this in mind, we now focus our attention on the $n$-th term
	\begin{gather}\label{Equation: connected n-point function contribution}
	\big(
	\Omega_\beta\circ\mathsf{R}^{\textrm{cl}}_{\lambda Q\{\chi_\mu\}}\big)^c
	\Bigg[
	FG\otimes\bigotimes_{\ell=1}^n (\lambda Q\{\dot{\chi}_\mu\})_{iu_\ell}
	\Bigg]\,.
	\end{gather}
	Since the state $\Omega_\beta\circ\mathsf{R}^{\textrm{cl}}_{\lambda Q\{\chi_\mu\}}$ is a quasi-free state with two-point function given by \eqref{Equation: two point function of KMS under the pull-back of classical Moller}, the term \eqref{Equation: connected n-point function contribution} can be expanded graphically as a sum of connected graphs, whose edges are associated with $\omega_{\lambda V\{\chi_\mu\}}$.
	Since $FG$ and $Q\{\chi_\mu\}$ are quadratic functionals, we can write \eqref{Equation: connected n-point function contribution} equivalently as a sum over $n$-permutations.
	Indeed for each $\sigma\in\wp_n$, the corresponding contribution to \eqref{Equation: connected n-point function contribution} is
	\begin{align*}
		\frac{(\lambda m_0^2)^n}{2}\Bigg[
		\int_{\mathbb{R}^{4(n+2)}}
		\textrm{d}Z\,
		\big[f(z_0)g(z_{n+1})+f(z_{n+1})g(z_0)\big]
		\omega_{\lambda Q\{\chi_\mu\}}(z_0,z_{\sigma_1})
		\omega_{\lambda Q\{\chi_\mu\}}(z_{n+1},z_{\sigma_{n}})\cdot\\
		\cdot
		\prod_{j=1}^{n-1}\omega_{\lambda Q\{\chi_\mu\}}(z_{\sigma_j},z_{\sigma_{j+1}})
		\prod_{\ell=1}^n
		(\dot{\chi}_\mu)_{iu_\ell}(z_\ell^0)
		\Bigg]\,,
	\end{align*}
	Here $z_0,\ldots,z_{n+1}\in M$ are arbitrary but fixed points of $M$ and $z^0_\ell$ denotes the time component of $z_\ell$ for $\ell\in\{1,\dots,n\}$.
	Notice that the symmetric term in $f,g$ gives the same contribution and cancels the factor $1/2$ which pops out from formula \eqref{Equation: definition of star-product}.

	For the sake of simplicity, we now provide the explicit computation in the case $n=1$, which will be generalized later.
	For $n=1$ the unique contribution in \eqref{Equation: connected n-point function contribution} is given by
	\begin{align}\label{Equation: contribution for G1}
		\lambda m_0^2
		\int_{\mathbb{R}^{12}}\textrm{d}z_0\textrm{d}z_1\textrm{d}z_2\,
		f(z_0)g(z_2)
		\omega_{\lambda V\{\chi_\mu\}}(z_0,z_1+iue^0)
		\omega_{\lambda V\{\chi_\mu\}}(z_2,z_1+iue^0)
		\dot{\chi}_\mu(z_1^0)\,,
	\end{align}
	Here $y+iue^0$ denotes the complex time translation of the point $y=(\underline{y},y^0)$ by $iue^0$ -- $e^0$ being the unit time vector field -- which we recall is well-defined once exploiting the analytic properties of $\omega_\beta$.
	
	Exploiting the explicit form of the two-point function $\omega_{\lambda V\{\chi_\mu\}}$ -- \textit{cf.} equation \eqref{Equation: two point function of KMS under the pull-back of classical Moller} -- and Lemma \ref{Lemma: useful limits} one finds that, in the limit $\mu\to+\infty$,
	\begin{gather*}
	\eqref{Equation: contribution for G1}\to
	\int_{\mathbb{R}^3}\frac{\textrm{d}k}{2\epsilon_\lambda}\,
	\sum_{\pm}
	\widehat{f}(\pm\epsilon_{\lambda},k)\widehat{g}(\mp\epsilon_{\lambda},-k)
	\frac{\lambda m_0^2}{\epsilon_\lambda+\epsilon}
	b_+(\beta,\epsilon)b_-(\beta,\epsilon)
	\,,
	\end{gather*}
	The final step is to recall \eqref{Equation: first beta derivative of Bose Einstein factor} so that the limit for $\mu\to+\infty$ of the term $n=1$ in \eqref{Equation: mu-dependent interacting thermal state} becomes
	\begin{gather*}
	n=1 \textrm{ term of \eqref{Equation: mu-dependent interacting thermal state}}\rightarrow_{\mu\to+\infty}
	\int_{\mathbb{R}^3}\frac{\textrm{d}k}{2\epsilon_\lambda}
	\sum_\pm
	\widehat{f}(\pm\epsilon_{\lambda},k)\widehat{g}(\mp\epsilon_{\lambda},-k)
	\frac{\beta\lambda m_0^2}{(\epsilon_\lambda+\epsilon)\epsilon}
	\partial_\beta b_\pm(\beta,\epsilon)\,.
	\end{gather*}
	The claim is that this formula can be generalized for all $n\geq 1$.
	Indeed, let consider the contribution to \eqref{Equation: connected n-point function contribution} at order $n\geq 1$.
	The combinatorial expansion gives
	\begin{align}
		\nonumber
		\eqref{Equation: connected n-point function contribution}=
		\sum_{\sigma\in\wp_n}
		(\lambda m_0^2)^n\Bigg[
		\int_{\mathbb{R}^{4(n+2)}}
		\textrm{d}Z\,
		f(z_0)g(z_{n+1})
		\omega_{\lambda Q\{\chi_\mu\}}(z_0,z_{\sigma_1})
		\omega_{\lambda Q\{\chi_\mu\}}(z_{n+1},z_{\sigma_{n}})\cdot\\
		\label{Equation: graph expansion of n-order contribution}
		\prod_{j=1}^{n-1}\omega_{\lambda V\{\chi_\mu\}}(z_{\sigma_j},z_{\sigma_{i+1}})
		\prod_{\ell=1}^n
		(\dot{\chi}_\mu)_{iu_\ell}(z_\ell^0)
		\Bigg]\,,
	\end{align}
	Let us focus on the contribution to \eqref{Equation: graph expansion of n-order contribution} from an arbitrary but fixed permutation $\sigma\in\wp_n$.
	Once the explicit expression \eqref{Equation: two point function of KMS under the pull-back of classical Moller} of $\omega_{\lambda V\{\chi_\mu\}}$ has been inserted into \eqref{Equation: graph expansion of n-order contribution}, one finds a sum of products of factors $b_\pm(\beta,\epsilon)$ with the corresponding modes.
	Notice that, since the limit $h\to 1$ has already been taken, there is a single integration over three momentum $k\in\mathbb{R}^3$.
	Invoking Lemma \ref{Lemma: useful limits} several terms in the sum above disappear.
	Actually, the non-vanishing contributions are those which contain, for each $\ell\in\{1,\ldots, n\}$, the factor $\big|T_k(z_\ell^0)\big|^2$ -- the other products would contain either a factor $\overline{T_k(z_\ell^0)^2}$ or a factor $T_k(z_\ell^0)^2$.
	These non-trivial terms can be computed as follows.
	Let be $j$ the number of descents of $\sigma$, that is, let $j\in\{1,\ldots,n\}$ be such that $\sigma\in\wp_{n,j}$.
	Then, in the adiabatic limit $\mu\to+\infty$, the non-trivial contribution to \eqref{Equation: graph expansion of n-order contribution} obtained from $\sigma$ is
	\begin{align}\label{Equation: contribution of a graph in the subclass n,k}
		\int_{\mathbb{R}^3}\frac{\textrm{d}k}{2\epsilon_\lambda}
		\sum_{\pm}
		\bigg[\frac{\lambda m_0^2}{(\epsilon_\lambda+\epsilon)}
		\bigg]^n
		b_\pm(\beta,\epsilon)^{n+1-j}b_\mp(\beta,\epsilon)^j
		\widehat{f}(\pm\epsilon_\lambda,k)\widehat{g}(\mp\epsilon_\lambda,-k)\,.
	\end{align}
	Thus, the contribution as $\mu\to+\infty$ arising from a $n$-permutation $\sigma\in\wp_n$ depends uniquely on its subclass $\wp_{n,j}$, \textit{i.e.} on the number of descends of $\sigma$.
	With this in mind we may rewrite the limit for $\mu\to+\infty$ of \eqref{Equation: graph expansion of n-order contribution} as follows:
	\begin{align}
		\nonumber
		\eqref{Equation: graph expansion of n-order contribution}
		\stackrel{\mu\to+\infty}{\rightarrow}&
		\int_{\mathbb{R}^3}\frac{\textrm{d}k}{2\epsilon_\lambda}
		\sum_{\pm}
		\widehat{f}(\pm\epsilon_\lambda,k)\widehat{g}(\mp\epsilon_\lambda,-k)
		\bigg[\frac{\lambda m_0^2}{(\epsilon_\lambda+\epsilon)}
		\bigg]^n
		\sum_{j=1}^n
		c_{n,j}
		b_\pm(\beta,\epsilon)^{n+1-j}b_\mp(\beta,\epsilon)^j
		\\
		\label{Equation: limit of the graph expansion of the n-order contribution}
		&=
		\int_{\mathbb{R}^3}\frac{\textrm{d}k}{2\epsilon_\lambda}\sum_{\pm}
		\widehat{f}(\pm\epsilon_\lambda,k)\widehat{g}(\mp\epsilon_\lambda,-k)
		\bigg[\frac{-\lambda m_0^2}{(\epsilon_\lambda+\epsilon)\epsilon}
		\bigg]^n
		\partial_\beta^n b_\pm(\beta,\epsilon)\,,
	\end{align}
	where we exploited the equality $|\wp_{n,k}|=c_{n,k}$ and equation \eqref{Equation: n-order beta derivative of Bose Einstein factor} as well as the symmetry property $c_{n,j}=c_{n,n+1-j}$.

	Summing up, we have computed the limit as $\mu\to+\infty$ of the $n$-order term appearing in \eqref{Equation: mu-dependent interacting thermal state}.
	Indeed, since \eqref{Equation: limit of the graph expansion of the n-order contribution} does not depends on $U=(u_1,\ldots,u_n)$, the integral over the simplex $\beta S_n$ would simply provide a factor $\beta^n(n!)^{-1}$.
	We thus find
	\begin{align*}
		\lim_{\mu\to+\infty}
		\Omega_{\beta,\lambda Q\{\chi_\mu\}}\big[
		\mathsf{R}_{\lambda Q\{\chi_\mu\}}(F)&\star_\beta\mathsf{R}_{\lambda Q\{\chi_\mu\}}(G)
		\big]=
		\omega_{\textrm{\textsc{ad},cl}}(f,g)\\&+
		\sum_{n\geq 1}
		\int_{\mathbb{R}^3}\frac{\textrm{d}k}{2\epsilon_\lambda}
		\sum_{\pm}
		\widehat{f}(\pm\epsilon_\lambda,k)\widehat{g}(\mp\epsilon_\lambda,-k)
		\frac{1}{n!}
		\bigg[\frac{\beta\lambda m_0^2}{(\epsilon_\lambda+\epsilon)\epsilon}
		\bigg]^n
		\partial_\beta^n b_\pm(\beta,\epsilon)\\&=
		\int_{\mathbb{R}^3}\frac{\textrm{d}k}{2\epsilon_\lambda}
		\sum_{\pm}
		\widehat{f}(\pm\epsilon_\lambda,k)\widehat{g}(\mp\epsilon_\lambda,-k)
		b_\pm\bigg[
		\beta+\frac{\beta\lambda m_0^2}{(\epsilon_\lambda+\epsilon)\epsilon},\epsilon
		\bigg]\,,
	\end{align*}
	where we used the explicit form \eqref{Equation: two point function of KMS under classical adiabatic limit} of $\omega_{\textrm{\textsc{ad},cl}}(f,g)$.
	It is then a simple computation to check that
	$b_\pm(\beta+\beta\lambda m_0^2[(\epsilon_\lambda+\epsilon)\epsilon]^{-1},\epsilon)=b_\pm(\beta,\epsilon_\lambda)$.

	The general case for $F,G\in\mathcal{P}_{\textrm{loc}}$ is treated analogously.
	Using relation \eqref{Equation: PPA} one reduces to the state $\Omega_\beta\circ\mathsf{R}^{\textrm{cl}}_{\lambda Q\{\chi_\mu\}}$ applied on local observables $\gamma_{\lambda Q\{\chi_\mu\}}(F), \gamma_{\lambda Q\{\chi_\mu\}}(G)$. Notice that $\gamma_{\lambda Q\{\chi_\mu\}}(F)=F+F'$, where $F'$ is a local functional with less fields than $F$.
	The combinatorial expansion exploited above still applies and the combinatorics reproduces the usual Wick formula for a quasi-free state.
	The thesis follows.
\end{proof}
\begin{remark}
One may wonder about the massless case $m=0$.
In this case the proof of Theorem \ref{Theorem: main theorem} is affected by several infrared divergences, that is the integral over three momentum $k\in\mathbb{R}^3$ is divergent due to the singular behaviour of the Bose-Einstein factor at $k=0$.
For example, the contribution \eqref{Equation: contribution of a graph in the subclass n,k} is divergent due to the presence of the product $b_{\pm}(\beta,\epsilon)^{n+1-j}b_\mp(\beta,\epsilon)^j\simeq|k|^{-n-1}$.
This singular behaviour can be understood by observing that the expansion $b_\pm(\beta,\epsilon_\lambda)=b_\pm\big(\beta+\beta\lambda m_0^2\big[\epsilon(\epsilon+\epsilon_\lambda)\big]^{-1},\epsilon\big)$ becomes singular in the massless case.
\end{remark}

\section{Non-equilibrium steady state}\label{Section: NESS}
In this section we compare the result obtained in Theorem \ref{Theorem: main theorem} with those obtained in \cite{DrFaPi17} where a non-equilibrium steady state (NESS) \cite{Ru00} $\Omega_{\textrm{\textsc{ness}}}$ was built out of $\Omega_{\beta,\lambda Q\{\chi\}}$ with an ergodic mean -- \textit{cf.} equation \eqref{Equation: definition of the NESS}.
In the case of a quadratic perturbation $Q$ as in \eqref{Equation: quadratic perturbation}, the steps of the proof of Theorem \ref{Theorem: main theorem} can be carried out also in this latter case, leading to some simplification.
However, such simplification would not suffice to sum the perturbative series.

As a preliminary result we state the following lemma.
\begin{lemma}\label{Lemma: useful limits for NESS}
	Let $\chi\in C^\infty(\mathbb{R})$ be such that \eqref{Equation: chi-properties} holds true.
	For any $k\in\mathbb{R}^3$, let $T_k(t):=T_{k,\mu=1}(t)$ be the modes defined as in \eqref{Equation: interpolating modes}.
	Then, for all $k\in\mathbb{R}^3$, there exist $A_\pm=A_\pm(k)\in\mathbb{C}$ such that, for all $t_1,t_2\in\mathbb{R}$,
	\begin{subequations}\label{Equation: useful limits for NESS}
	\begin{gather}
		\label{Equation: useful limits for NESS TT}
		\lim_{t\to+\infty}
		\frac{1}{t}\int_0^t\textrm{d}\tau\,
		T_k(t_1+\tau)T_k(t_2+\tau)=
		\frac{A_+A_-}{2\epsilon_\lambda}
		\sum_\pm
		e^{\mp i\epsilon_\lambda(t_1-t_2)}
		\,,\\
		\label{Equation: useful limits for NESS T barT}
		\lim_{t\to+\infty}
		\frac{1}{t}\int_0^t\textrm{d}\tau\,
		T_{k}(t_1+\tau)\overline{T_{k}(t_2+\tau)}=
		\frac{1}{2\epsilon_\lambda}
		\sum_\pm
		|A_\pm|^2e^{\mp i\epsilon_\lambda(t_1-t_2)}
		\,.
	\end{gather}
	\end{subequations}
\end{lemma}
\begin{proof}
	First, we recall from \cite[Lemma 5.1]{DaDr16}, see also \cite[Appendix D]{DrHaPi16} that the modes $T_k$ defined in \eqref{Equation: interpolating modes} are uniformly bounded, namely
	\begin{gather}\label{Equation: uniform bound on interpolating modes}
	|T_k(t)|\leq
	\frac{1}{2\epsilon_\lambda}=
	\frac{1}{2\sqrt{|k|^2+m^2+\lambda m_0^2}}\,.
	\end{gather}
	For $\tau\geq0$, $T_k(\tau)$ satisfies equation \eqref{Equation: interpolating modes}, where $\epsilon_{\mu=1}(k,\tau)=\epsilon_\lambda(k)$.
	Hence we may write $T_k(\tau)$ as
	\begin{align}\label{Equation: modes expression for large positive times}
		T_k(\tau)=
		\frac{1}{\sqrt{2\epsilon_\lambda}}
		\sum_{\pm}
		A_\pm e^{\mp i\epsilon_\lambda\tau}\,,\qquad A_\pm\in\mathbb{C}\,.
	\end{align}
	Let $t_1,t_2\in\mathbb{R}$ be arbitrary but fixed and set
	\begin{align}
		\tau_{\textrm{min}}:=\inf
		\{
		\tau\in\mathbb{R}|\;
		t_\ell+\tau\geq 0\;
		\forall\ell\in\{1,2\}
		\}\,.
	\end{align}
	For $\tau\geq\tau_{\textrm{min}}$ we may compute
	\begin{subequations}\label{Equation: asymptotic expression for translated interacting modes}
	\begin{align}
		\label{Equation: asymptotic expression for translated interacting modes TTbar}
		T_k(t_1+\tau)\overline{T_k(t_2+\tau)}&=
		\frac{1}{2\epsilon_\lambda}\sum_{\pm}
		\bigg[|A_\pm|^2 e^{\mp i\epsilon_\lambda(t_1-t_2)}\bigg]+
		2\Re\big(A_+\overline{A_-}e^{- i\epsilon_\lambda(2\tau+t_1+t_2)}\big)
		\\
		T_k(t_1+\tau)T_k(t_2+\tau)&=
		\frac{1}{2\epsilon_\lambda}
		\sum_{\pm}\bigg[
		A_\pm^2 e^{\mp i\epsilon_\lambda(2\tau+t_1+t_2)}+
		A_+A_-e^{\mp i\epsilon_\lambda(t_1-t_2)}
		\bigg]
		\,.
	\end{align}
	\end{subequations}
	The limits \eqref{Equation: useful limits for NESS} can be computed exploiting \eqref{Equation: uniform bound on interpolating modes}: considering \eqref{Equation: useful limits for NESS T barT} we have
	\begin{align*}
		\lim_{t\to+\infty}
		\int_0^t\frac{\textrm{d}\tau}{t}\,
		T_k(t_1+\tau)\overline{T_k(t_2+\tau)}=
		\lim_{t\to+\infty}
		\int_{\tau_{\textrm{min}}}^t\frac{\textrm{d}\tau}{t}\,
		T_k(t_1+\tau)\overline{T_k(t_2+\tau)}=
		\frac{1}{2\epsilon_\lambda}\sum_{\pm}
		A_\pm^2 e^{\mp i\epsilon_\lambda(t_1-t_2)}\,,
	\end{align*}
	where in the second equality we used \eqref{Equation: asymptotic expression for translated interacting modes TTbar}.
	A similar computation can be carried out for \eqref{Equation: useful limits for NESS TT}.
	The thesis follows.
\end{proof}
\begin{remark}
\textit{(i)}
The coefficients $A_\pm$ are subjected to the condition $|A_+|^2-|A_-|^2=1$ which ensures the Wronskian condition \eqref{Equation: Wronski condition} for the modes $T_k$.\\
\textit{(ii)}
The proof of Lemma \eqref{Lemma: useful limits for NESS} still holds true if one replaces the ergodic mean over $\tau\in (0,+\infty)$ with an ergodic mean over the real axis.
In this latter case equation \eqref{Equation: useful limits for NESS} acquires additional terms, proportional to the modes $\exp\big[\pm i\epsilon(t_1-t_2)\big]$.
\end{remark}
As an immediate consequence of Lemma \ref{Lemma: useful limits for NESS} we compute the ergodic limit of the state $\Omega_\beta\circ\mathsf{R}^{\textrm{cl}}_{\lambda Q\{\chi\}}$.
\begin{corollary}\label{Corollary: ergodic limit of KMS state under classical Moller operator}
	The limit $t\to+\infty$ of the sequence of states defined by $A\mapsto t^{-1}\int_0^t\textrm{d}s\,\big(\Omega\circ\mathsf{R}^{\textrm{cl}}_{\lambda Q\{\chi\}}\big)[\alpha_s(A)]$ exists and defines a quasi-free state $\Omega_{\textrm{\textsc{ness},cl}}$ whose two-point function is given by
	\begin{align}\label{Equation: ergodic limit of KMS state under classical Moller operator}
		\omega_{\textrm{\textsc{ness},cl}}(f,g):=
		\int_{\mathbb{R}^3}
		\frac{\textrm{d}k}{2\epsilon_\lambda}
		\sum_\pm
		c_\pm(\beta,k)
		\widehat{f}(\pm\epsilon_\lambda,k)\widehat{g}(\mp\epsilon_\lambda,-k)\,,
	\end{align}
	where 
	\begin{align}
		c_+(\beta,k):=
		\sum_\pm
		b_\pm(\beta,\epsilon)|A_\pm(k)|^2\,,\qquad
		c_-(\beta,k):=
		\sum_\pm
		b_\pm(\beta,\epsilon)|A_\mp(k)|^2\,.
	\end{align}
\end{corollary}
\begin{remark}
	Notice that the CCR relations for the state $\omega_{\textrm{\textsc{ness},cl}}$ are a direct consequence of the relations $b_+-b_-=1=|A_+|^2-|A_-|^2$.
\end{remark}

We now follows the steps of the proof of Theorem \ref{Theorem: main theorem} in the case of $\Omega_{\textrm{\textsc{ness}}}$.
In particular, let $F,G\in\mathcal{A}$ be linear functionals, namely $F(\phi):=\int_Mf\phi$ and $G(\phi):=\int_Mg\phi$ for $f,g\in C^\infty_{\textrm{c}}(M)$.
We shall compute
\begin{align}\label{Equation: two-point function of the ergodic limit of the interacting KMS state}
	\Omega_{\textrm{\textsc{ness}}}
	\big(
	\mathsf{R}_{\lambda Q\{\chi\}}(F)\star_\beta
	\mathsf{R}_{\lambda Q\{\chi\}}(G)
	\big):=
	\lim_{t\to+\infty}\frac{1}{t}
	\int_0^t\textrm{d}s\;
	\Omega_{\beta,\lambda Q\{\chi\}}
	\big[
	\tau_s\big(\mathsf{R}_{\lambda Q\{\chi\}}(F)\big)\star_\beta
	\tau_s\big(\mathsf{R}_{\lambda Q\{\chi\}}(G)\big)
	\big]\,.
\end{align}
The existence of the limit $t\to +\infty$ in the sense of formal power series in $\lambda$ has already been proved in \cite{DrFaPi17}.
As in \eqref{Equation: mu-dependent interacting thermal state}, we exploit the perturbative series \eqref{Equation: interacting thermal state} in order to compute the integrand
\begin{align}
	\nonumber
	\Omega_{\beta,\lambda Q\{\chi\}}
	\big[
	\tau_s\big(\mathsf{R}_{\lambda Q\{\chi\}}(F)\big)\star
	\tau_s\big(\mathsf{R}_{\lambda Q\{\chi\}}(G)\big)
	\big]&=
	\omega_{\lambda Q\{\chi\}}(f_s,g_s)\\
	\label{Equation: tau-translated interacting thermal state}
	&+
	\sum_{n\geq 1}
	(-1)^n\int_{\beta S_n}
	\big[
	\Omega_\beta\circ\mathsf{R}^{\textrm{cl}}_{\lambda Q\{\chi\}}
	\big]^c
	\bigg[
	F_s G_s
	\otimes\bigotimes_{\ell=1}^nQ\{\dot{\chi}\}_{iu_\ell}
	\bigg]\,.
\end{align}
Thanks to Corollary \ref{Corollary: ergodic limit of KMS state under classical Moller operator} one finds that, once integrated in $s$, the first term in the right-hand side of \eqref{Equation: tau-translated interacting thermal state} converges to $\omega_{\textrm{\textsc{ness},cl}}(f,g)$.
Similarly to the proof of Theorem \ref{Theorem: main theorem} we focus on the expansion of the $n$-th term
\begin{align}\label{Equation: connected n-point function contribution for ergodic mean}
	\big[
	\Omega_\beta\circ\mathsf{R}^{\textrm{cl}}_{\lambda Q\{\chi\}}
	\big]^c
	\bigg(
	F_s G_s
	\otimes\bigotimes_{\ell=1}^nQ\{\dot{\chi}\}_{iu_\ell}
	\bigg)\,.
\end{align}
Once again, this contribution is the sum over $n$-permutations, in particular
\begin{align}
	\label{Equation: graph expansion of n-order contribution for ergodic mean}
	\eqref{Equation: connected n-point function contribution for ergodic mean}&=
	\sum_{\sigma\in\wp_n}
	(\lambda m_0^2)^n
	\int_{\mathbb{R}^{4(n+2)}}
	\textrm{d}Z
	f(z_0)g(z_{n+1})
	\omega_{\lambda Q\{\chi\}}(z_0+s e^0,z_{\sigma_1})
	\omega_{\lambda Q\{\chi\}}(z_{n+1}+s e^0,z_{\sigma_n})\Xi_\sigma(\hat{Z},U)\,,\\&
	\Xi_\sigma(\hat{Z},U):=
	\prod_{j=1}^{n-1}\omega_{\lambda Q\{\chi\}}(z_{\sigma_j},z_{\sigma_{j+1}})
	\prod_{\ell=1}^n
	\dot{\chi}_{iu_\ell}(z_\ell^0)\,,\qquad\hat{Z}:=(z_1,\ldots,z_n)\,.
\end{align}
Exploiting equation \eqref{Equation: two point function of KMS under the pull-back of classical Moller} (for $\mu=1$) one reduces the previous expression \eqref{Equation: graph expansion of n-order contribution for ergodic mean} to an integral in three momentum $k\in\mathbb{R}^3$.
Considering the ergodic mean of the contribution to \eqref{Equation: graph expansion of n-order contribution for ergodic mean} of a single $n$-permutation $\sigma\in\wp_n$ and applying Lemma \ref{Lemma: useful limits for NESS}, one finds
\begin{align}
	\begin{array}{l}
	\textrm{ergodic mean of}\\
	\textrm{the $\sigma$-contribution to } \eqref{Equation: graph expansion of n-order contribution for ergodic mean}
	\end{array}\to_{t\to+\infty}
	(\lambda m_0^2)^n\int_{\mathbb{R}^3}\frac{\textrm{d}k}{2\epsilon_\lambda}\Bigg[
	\sum_{\pm}
	a_{\pm,\sigma}(k,U)
	\widehat{f}(\pm\epsilon_\lambda,k)
	\widehat{g}(\mp\epsilon_\lambda,-k)
	\Bigg]\,,
\end{align}
where we defined
\begin{align*}
	a_{\pm,\sigma}(k,U):=\int_{\mathbb{R}^{4n}}\textrm{d}\hat{Z}\bigg[
	b_+(\epsilon)^2A_+A_-\overline{T_k(z_{\sigma_1})T_k(z_{\sigma_n})}+
	b_-(\epsilon)^2\overline{A_+A_-}T_k(z_{\sigma_1})T_k(z_{\sigma_n})\\+
	b_+(\epsilon)b_-(\epsilon)\bigg(
	|A_\pm|^2\overline{T_k(z_{\sigma_1})}T_k(z_{\sigma_n})+
	|A_\mp|^2T_k(z_{\sigma_1})\overline{T_k(z_{\sigma_n})}
	\bigg)
	\bigg]\Xi_\sigma(\hat{Z},U)\,.
\end{align*}
Hence, in spite of the fact that the ergodic limit of each term present in \eqref{Equation: graph expansion of n-order contribution for ergodic mean} can be computed, the resulting limit appears to depend in a quite complicate way on the chosen graph $\sigma\in\wp_n$.
This fact spoils the chance to infer a closed form for the ergodic limit of the series \eqref{Equation: two-point function of the ergodic limit of the interacting KMS state}.

\paragraph{Acknowledgements.}
The author is grateful to
Claudio Dappiaggi,
Federico Faldino,
Klaus Fredenhagen,
Thomas-Paul Hack and
Nicola Pinamonti,
for enlightening discussions and comments on a preliminary version of this paper.
This work was supported by the National Group of Mathematical Physics (GNFM-INdAM).

\end{document}